\documentclass[11pt]{article}
\pdfoutput = 1
\usepackage[utf8]{inputenc}

\usepackage[top=1in,bottom=1in,left=1in,right=1in]{geometry}
\makeatletter \g@addto@macro\@floatboxreset\centering \makeatother
\usepackage{color}
\definecolor{darkgreen}{rgb}{0,0.5,0}
\definecolor{darkblue}{rgb}{0,0,0.6}
\definecolor{purple}{rgb}{0.4,.2,0.7}
\usepackage[noEucal]{main}
\usepackage{aas_macros,empheq,fancybox,graphicx,multirow}
\usepackage{aas_macros}
\usepackage{dsfont}
\usepackage{esint}
\usepackage{float}
\usepackage{appendix}
\usepackage{mathabx}
\usepackage{hhline}
\usepackage{tensor} 
\usepackage{footmisc}
\usepackage{float}
\usepackage{caption}
\newtheorem{definition}{Definition}

\newtheorem{proposition}{Proposition}[section]

\newtheorem{theorem}{Theorem}[section]

\newtheorem{lemma}{Lemma}[section]

\newenvironment{proof}{\smallskip\noindent\emph{Proof.}\hspace{1pt}}%
{\hspace{-5pt}{\nobreak\quad\nobreak\hfill\nobreak$\square$\vspace{8pt}%
		\par}\smallskip\goodbreak}

\newcommand{\hnabla}{\widehat{\nabla}}

\newcommand{\bm}{\begin{align*}}
\newcommand{\enm}{\end{align*}}

\newcommand{\bespeq}{\begin{equation}\begin{split}}
\newcommand{\espeq}{\end{split}\end{equation}}

\def\S {S_{u,\underline{u}}}

\def\p{\psi}

\def\p{\psi}

\renewcommand{\div}{\mbox{div }}

\begin{document}

\thispagestyle{empty}

\begin{center}
    ~
    \vskip10mm

     {\LARGE  {\textsc{A new conformal quasi-local energy in general relativity}}}
    \vskip10mm
    
Puskar Mondal$^{a,b}$ and  Shing-Tung Yau$^{a,b,c}$ \\
    \vskip1em
    {\it
        $^a$ Center of Mathematical Sciences and Applications, Harvard University, Cambridge, Massachusetts 02138, USA\\ \vskip1mm
$^b$ Department of Mathematics, Harvard University, Cambridge, Massachusetts 02138, USA\\ \vskip1mm
$^{c}$ Yau Mathematical Sciences Centre, Tsinghua University\\ \vskip1mm 30 Shuangqing Rd, Beijing 100190, China
         \vskip1mm
    }
    \vskip5mm
    \tt{yau@math.harvard.edu, puskar$\_$mondal@fas.harvard.edu}
\end{center}
\vspace{10mm}

\begin{abstract}
\noindent We construct new conserved quasi-local energies in general relativity using the formalism developed by \cite{CWY}. In particular, we use the optimal isometric embedding defined in \cite{yau,yau1} to transplant the conformal Killing fields of the Minkowski space back to the $ 2-$ surface of interest in the physical spacetime. For an asymptotically flat spacetime of order $1$, we show that these energies are always finite. Their limit as the total energies of an isolated system is evaluated and a conservation law under Einsteinian evolution is deduced.

\end{abstract}
\pagebreak

\setcounter{tocdepth}{2}
{\hypersetup{linkcolor=black}
\small
\tableofcontents
}

\section{Introduction}
A local notion of energy is absent in the context of pure gravity due to the existence of a geodesic normal coordinate chart at each point on a semi-Riemannian manifold or in physical terms the equivalence principle. Since the energy content in a physical spacetime would roughly correspond to a measure of being different from the topologically trivial flat spacetime, one may attempt to define a local energy density of pure gravity by constructing invariants of the Weyl curvature. One such entity would be the Bel-Robinson stress-energy tensor. Unfortunately, one can not create dimensionally consistent energy from the Bel-Robinson tensor at the classical level without introducing an additional length scale to the system. Therefore, it becomes essential to define a notion of energy in pure gravity over an extended region of spacetime. There exists the global notion of mass such as the ADM mass for asymptotically flat spacetimes and its desired physical properties such as positivity are proven in the celebrated work of Schoen and Yau \cite{schoen1979proof,schoen1981proof} and later by Witten \cite{witten1981new,parker1982witten}. In addition, there also exists the notion of mass for asymptotically hyperbolic manifolds i.e., for which the Cauchy slices are asymptotically hyperbolic \cite{chrusciel2003mass, chrusciel2019hyperbolic}. Apart from these global notions of energy, it becomes essential to define quasi-local energy in certain physical situations (e.g., to understand the gravitational radiation of binary black hole merger). 

\noindent Penrose \cite{penrose1982some} listed a set of major open problems in mathematical general relativity in 1982 that contained the following problem: \textit{find a quasi-local definition of energy-momentum in the context of pure gravity}. To be more precise, one asks to evaluate the energy content in a spacelike region bounded by a topological $2-$sphere. Some desirable physical properties of this definition of energy should include non-negativity under the assumption of an appropriate energy condition on the spacetime and rigidity i.e., it should vanish for a topological $2-$sphere embedded in the Minkowski space. In addition, it should encode the information about pure gravitational energy (Weyl curvature effect) as well as the stress-energy tensor of any source fields present in the spacetime. Based on a Hamilton-Jacobi analysis, Brown-York \cite{brown1992quasilocal, brown1993quasilocal} and Liu-Yau \cite{liu2003positivity,liu2006positivity} defined a quasi-local mass by isometrically embedding the $2-$surface into the reference Euclidean $3-$ space and comparing the extrinsic geometry (the formulation relied on the embedding theorem of Pogorelov \cite{pogorelov1952regularity} i.e., the topological $2-$sphere needed to posses everywhere non-negative sectional curvature). However, \cite{murchadha2004comment} discovered surfaces in Minkowski space that do have strictly positive Brown-York and Liu-Yau mass. It appeared that these formulations lacked a prescription of momentum information. This led Wang and Yau \cite{yau,yau1} to define the most consistent notion of the quasi-local mass associated with a space-like topological $2-$sphere (we should mention that the second author together with Alae and Khuri recently provided a new definition of quasi-local mass bounded by a class of surfaces verifying certain topological conditions in \cite{alaee2023quasi}, we should also mention that very recently John Lott \cite{lott} defined a new notion of quasi-local mass using spinor method). It involves isometric embedding of the topological $2-$sphere bounding a space-like domain in the physical spacetime satisfying dominant energy condition (energy can not flow into a past light cone of an arbitrary point in spacetime; essentially finite propagation speed) into the Minkowski spacetime instead of Euclidean $3-$space. This formulation relies on a weaker condition on the sectional curvature of the $ 2-$ surface of interest and solvability of Jang's equation \cite{jang1978positivity} with prescribed Dirichlet boundary data. The Wang-Yau quasi-local mass is then defined as the infimum of the Wang-Yau quasi-local energy among all physical observers that is found by solving an optimal isometric embedding equation. This Wang-Yau quasi-local mass possesses several good properties that are desired on physical ground. This mass is strictly positive for $2-$surfaces bounding a space-like domain in a curved spacetime that satisfies the dominant energy condition and it identically vanishes for any such $2-$surface in Minkowski spacetime. In addition, it coincides with the ADM mass at space-like infinity \cite{wang2010limit} and Bondi-mass at null infinity \cite{chen2011evaluating} and reproduces the time component of the Bel-Robinson tensor (a pure gravitational entity) together with the matter stress-energy at the small sphere limit, that is when the $2-$sphere of interest is evolved by the flow of its null geodesic generators and the vertex of the associated null cone is approached \cite{chen2018evaluating}. In addition, explicit conservation laws were also discovered at the asymptotic infinity \cite{CWY}.  

\noindent In addition to the charges (energy, momentum, angular momentum, and center of mass) defined by means of a Minkowskian Killing vector field transplanted onto the physical spacetime through the isometric embedding of a topological $2-$sphere, one may ask to construct charges by means of conformal Killing vector fields. Recall that in the original proof of the stability of the Minkowski space by Christodoulou and Klainerman \cite{christodoulou1993global}, the conformal vector fields played a crucial role in constructing coercive entities that controlled appropriate norms of the dynamical variables. In particular, the scaling vector field $S=X^{\mu}\partial_{\mu}$ (in the usual rectangular global chart $\{X^{\mu}=(t=X^{0},X^{1},X^{2},X^{3})\}$) and the inversion generator $K=\eta(X,X)\partial_{t}+2X^{0}X^{\alpha}\partial_{\alpha}$ were used (note that these are exactly conformal Killing vector fields for Minkowski space but in the stability problem where small perturbations of Minkowski space was considered, they served as approximate conformal Killing vector fields i.e., the trace-free part of their deformation tensors are small in an appropriate function space setting). The reason behind using these energies was that upon contraction with the Bel-Robinson tensor, these vector fields generate coercive entities that control appropriate \textit{weighted} norms of the curvature. Therefore, in an asymptotically flat spacetime, where usual Sobolev norms are replaced by weighted ones, these entities were the natural ones to define. One can ask a similar question: Can one construct a quasi-local entity in the sense of Chen-Wang-Yau \cite{WY} for the conformal Killing vector fields of the Minkowski space?        

\noindent In this article we consider the scaling vector field $S=x^{\mu}\partial_{\mu}$ and construct a quasi-local entity. We study its properties and prove a conservation law for asymptotically flat spacetime of order one. In particular, we use the definition provided by Chen-Wang-Yau \cite{CWY} that was further motivated by the earlier study of Wang-Yau \cite{yau,yau1}. Let us introduce the notion of generalized energies defined by Chen-Wang-Yau \cite{CWY}. Consider a topological $2-$sphere $\Sigma$ in the physical spacetime $(M,\widehat{g})$ that bounds a spacelike topological ball. Let us assume that the mean curvature vector $\mathbf{H}$ of $\Sigma$ is space-like. Let $\mathbf{J}$ be the reflection of $\mathbf{H}$ through the future outgoing light cone in the normal bundle of $\Sigma$. The data that Wang and Yau use to define the quasi-local energy is the triple ($\sigma,|\mathbf{H}|_{\hat{g}},\alpha_{\mathbf{H}}$) on $\Sigma$, where $\sigma$ is the induced metric on $\Sigma$ by the Lorentzian metric $\hat{g}$ on the physical spacetime $(M,\widehat{g})$, $|\mathbf{H}|_{\hat{g}}$ is the Lorentzian norm of $\mathbf{H}$, and $\alpha_{\mathbf{H}}$ is the connexion $1-$form of the normal bundle with respect to the mean curvature vector $\mathbf{H}$ and is defined as follows 
\begin{eqnarray}
\label{eq:connection}
\alpha_{\mathbf{H}}(X):=\hat{g}(\nabla[\hat{g}]_{X}\frac{\mathbf{J}}{|\mathbf{H}|},\frac{\mathbf{H}}{|\mathbf{H}|}).
\end{eqnarray}
Choose a basis pair ($e_{3},e_{4}$) for the normal bundle of $\Sigma$ in the spacetime that satisfy $\hat{g}(e_{3},e_{3})=1,\hat{g}(e_{4},e_{4})=-1$, and $\hat{g}(e_{3},e_{4})=0$. Now embed the $2-$surface $\Sigma$ isometrically into the Minkowski space with its usual metric $\eta$ i.e., the embedding map $X: x^{a}\mapsto X^{\mu}(x^{a})$ satisfies $\sigma(\frac{\partial}{\partial x^{a}},\frac{\partial}{\partial x^{b}})=\eta(\frac{\partial X}{\partial x^{a}},\frac{\partial X}{\partial x^{b}})$, where $\{x^{a}\}_{a=1}^{2}$ are the coordinates on $\Sigma$. Now identify a basis pair ($e_{30},e_{40}$) in the normal bundle of $X(\Sigma)$ in the Minkowski space that satisfies the exact similar property as $(e_{3},e_{4})$. In addition the time-like unit vector $e_{4}$ is chosen to be future directed i.e., $\hat{g}(e_{4},\partial_{t})<0$. Similarly identify the corresponding entities $\mathbf{H}_{0}, \alpha_{\mathbf{H}_{0}}$ of the image $X(\Sigma)$ in the Minkowski space. Let $\tau:=-\langle X,\partial_{t}\rangle_{\eta}$, a function on $\Sigma$ be the time function of the embedding. Technically, the isometric embedding condition only provides the solution of the three embedding functions out of the required four. The fourth embedding function is obtained through solving a fourth order elliptic PDE. Such embedding is entitled as the \textit{optimal} isometric embedding. Chen-Wang-Yau \cite{CWY} provided the following definition
\begin{definition}
\label{def1}
The quasi-local entity of a membrane $\Sigma$ with respect to an optimal isometric embedding $(X,T_{0})$ and a conformal Killing field $K$ of the Minkowski space is 
\begin{eqnarray}
E(\Sigma,X,T_{0},K):=-\frac{1}{8\pi}\int_{\Sigma}\left(\langle K,T_{0}\rangle \rho+K^{\perp}\cdot J\right)\mu_{\Sigma},
\end{eqnarray}
where the current vector field $J:=J^{a}\frac{\partial}{\partial x^{a}},~a=1,2$ is defined as follows 
\begin{eqnarray}
J:=\rho\nabla\tau-\nabla[\sinh^{-1}\frac{\rho\Delta \tau}{|H_{0}||H|}]-\alpha_{H_{0}}+\alpha_{H}
\end{eqnarray}
and $\rho$ reads 
\begin{eqnarray}
 \rho=\frac{\sqrt{|\mathbf{H}_{0}|^{2}+\frac{(\Delta\tau)^{2}}{1+|\nabla\tau|^{2}}}-\sqrt{|\mathbf{H}|^{2}+\frac{(\Delta\tau)^{2}}{1+|\nabla\tau|^{2}}}}{\sqrt{1+|\nabla\tau|^{2}}}   
\end{eqnarray}
\end{definition}
Essentially, this definition uses the isometric embedding to transplant the Killing vector field from the Minkowski space onto the physical spacetime $(M,\widehat{g})$ which does not necessarily have a Killing vector field. 
Here we use this definition to construct a quasi-local entity associated with the conformal Killing vector field the scaling vector field $S=X^{\mu}\partial_{\mu}$ of the Minkowski space. The same idea of transplanting the conformal Killing field $S$ onto the physical spacetime via isometric embedding is utilized. Regarding this new quasi-local entity we prove two theorems. The first theorem concerns its finiteness property in an vacuum asymptotically flat spacetime. The second theorem establishes a conservation law at the asymptotic infinity. In order to state the theorem, we need the definition of a certain asymptotically flat initial data set of the Einstein's equations 
\begin{definition}
\label{def2}
$(M,g,k)$ is asymptotically flat initial data set of order one if there is a compact subset $C$ of $M$ such that $M-C$ is diffeomorphic to $\mathbb{R}^{3}-B$ and the diffeomorphism incduces a ``Cartesian" coordinate system $\{x^{i}\}_{i=1}^{3}$ on $M-C$ in which the following decay condition holds for the metric $g$ and the second fundamental for $k$
\begin{eqnarray}
 g_{ij}=\delta_{ij}+\frac{g_{ij}^{-1}}{r}+\frac{g_{ij}^{-2}}{r^{2}}+o(r^{-2}),   
\end{eqnarray}
where $r:=\sqrt{\sum_{i=1}^{3}(x^{i})^{2}}$
\end{definition}
In this article, we will be concerned with asymptotically flat data set of order one as defined above. 
\begin{theorem}[Finiteness property]
 Let $(M,g,k)$ be an asymptotically flat initial data set of order one that verifies the vacuum constraint equation,
 then the total quasi-local energy $E(\Sigma,X,T_{0},S)$ associated with the conformal Killing vector field $S=X^{\mu}\partial_{\mu}$ is finite.  
\end{theorem}
The main idea behind the proof of the theorem is to solve the optimal isometric embedding equation asymptotically through series expansion. Naively at the level of scales, one observes that there are $O(r)$ terms in the definition of the energy $E$. However, for asymptotically flat spacetime of order one in definition \ref{def2}, such terms will vanish leading to finiteness of the total energy. This requires very careful analysis of the optimal isometric embedding equation order by order in $r$ and several integration by parts identities. Once the total energy associated with the scaling vector field $S$ is proven to be finite, then one intends to understand its behavior under Einsteinian evolution. It turns out that one can relate the time evolution of the total energy to that of the linear momentum and the ADM energy. In particular recall the Wang-Yau quasi-local energy $e$ and Chen-Wang-Yau quasi-local linear momentum $p$
\begin{eqnarray}
e=\frac{1}{8\pi}\int_{\Sigma}(\rho T^{0}_{0}+J^{a}\nabla_{a}X^{0})\mu_{\Sigma}\\
p^{i}=\frac{1}{8\pi}\int_{\Sigma}(\rho T^{i}_{0}+J^{a}\nabla_{a}X^{i})\mu_{\Sigma}.
\end{eqnarray}
The following theorem relates the time evolution of the energy $E$ associated with the scaling vector field $S=X^{\mu}\partial_{\mu}$ and the energy momentum $(e,p)$ at spacelike infinity.
\begin{theorem}[Conservation property]
\label{theorem2}
Let $(M,g,k)$ be an asymptotically flat initial data set of order one (as in definition \ref{def2}) satisfying the vacuum constraint equation. Let $(M,g(t),k(t))$ be the solution of the Einstein evolution equations with initial condition $g(0)=g$ and $k(0)=k$. Assume that the lapse function $N$ and the shift vector field $Y$ verifies $N=1+O(r^{-1}), Y=\frac{Y^{-1}}{r}+O(r^{-2})$. Then the total energy $E$ associated with the scaling vector field $S=X^{\mu}\partial_{\mu}$ verifies the following conservation law at spacelike infinity 
\begin{eqnarray}
 \partial_{t}E(t)=\frac{p^{2}}{e}   
\end{eqnarray}    
\end{theorem}

\noindent This formula can be re-written in terms of the evolution of the centre of mass $C^{i}$ following the result of Chen-Wang-Yau \cite{CWY}. Recall that the following conservation law holds by the result of \cite{CWY}
\begin{eqnarray}
 \partial_{t}C^{i}(t)=\frac{p^{i}}{e}.   
\end{eqnarray}
Therefore, our result implies $\partial_{t}E(t)=p_{i}\partial_{t}C^{i}(t)$. A more detailed physical interpretation of this result is yet to be thought of.

\section{Evaluation of the energy $E$: solution of the Isometric embedding equations}
In this section, we evaluate the quasi-local entity associated with the conformal Killing vector field $S=X^{\mu}\partial_{\mu}$ as defined in \ref{def1}. The quasi-local entity reads 
\begin{eqnarray}
E(\Sigma,X,T_{0},K=S=X^{\alpha}\partial_{\alpha})=-\frac{1}{8\pi}\int_{\Sigma}\left(\eta_{\alpha\beta}X^{\alpha}T^{\beta}_{0}\rho+X^{\alpha}\nabla_{a}X^{\beta}\eta_{\alpha\beta}J^{a}\right)\mu_{\Sigma}.
\end{eqnarray}
We want to evaluate this on large spheres foliating $M-C$. There is a standard spherical coordinate system $(r,\theta,\varphi)$ on $M-C$ defined as 
\begin{eqnarray}
 x^{1}=r\sin\theta\cos\varphi,~x^{2}=r\sin\theta\sin\varphi,~x^{3}=r\cos\theta.   
\end{eqnarray}
On each level set $\Sigma_{r}$ of $r$ we can use $(\theta,\varphi)$ as coordinates to express geometric data we need in order to evaluate the quasi-local conserved entity. Note the following expansions in $r$
\begin{eqnarray}
\sigma_{ab}=r^{2}\widehat{\sigma}_{ab}+r\sigma^{1}_{ab}+\sigma^{0}_{ab}+o(1),\\
|H|=\frac{2}{r}+\frac{h^{-2}}{r^{2}}+\frac{h^{-3}}{r^{3}}+o(r^{-3}),\\
\alpha_{H}=\frac{\alpha^{-1}_{H}}{r}+\frac{\alpha^{-2}_{H}}{r^{2}}+o(r^{-2}).
\end{eqnarray}
Here $\widehat{\sigma}$ is the standard round metric on the unit $2-$sphere $\mathbb{S}^{2}$, $\sigma^{1}_{ab}$ and $\sigma^{0}_{ab}$ are symmetric $2-$tensors on $\mathbb{S}^{2}$, $h^{-2},h^{-3}$ are functions on $\mathbb{S}^{2}$, and $\alpha^{-1}_{H}$ and $\alpha^{-2}_{H}$ are $1-$forms on $\mathbb{S}^{2}$. Note that these data are sufficient to evaluate the quasi-local entity. First we need to isometrically embed $\Sigma_{r}$ into the Minkowski space. Let $i:\Sigma_{r}\to (\mathbb{R}^{1,3},\eta),~(\theta,\varphi)\mapsto (X^{0}(\theta,\varphi),X^{1}(\theta,\varphi),X^{2}(\theta,\varphi),X^{3}(\theta,\varphi))$ be the embedding map. Recall the isometric embedding equation 
\begin{eqnarray}
\label{eq:isometric}
\langle dX,dX\rangle_{\eta}=\sigma.
\end{eqnarray}
Once again we can attempt to solve this equation order by order in $r$ through the following asymptotic expansion 
\begin{eqnarray}
X^{0}(r)=(X^{0})^{0}+\frac{(X^{0})^{-1}}{r}+o(r^{-1})\\
X^{i}(r)=r\widehat{X}^{i}+(X^{i})^{0}+\frac{(X^{i})^{-1}}{r}+o(r^{-1}),
\end{eqnarray}
where $\widehat{X}^{i}$ are the usual $SO(3)$ elements, $(X^{0})^{0}, (X^{i})^{0},(X^{0})^{-1},(X^{i})^{-1}$ are independent of $r$. Now note that the isometric embedding equation \ref{eq:isometric} only solves for three functions. Assuming we can solve for $X^{1},X^{2},X^{3}$, we still need to solve for the time function. Let $T_{0}(r)$ denote the time vector field and $\tau(r):=-\langle X(r),T_{0}(r)\rangle$. Note that for large $r$, $T_{0}$ starts coinciding with the constant vector field and therefore one can write     
\begin{eqnarray}
T_{0}=(a^{0},a^{1},a^{2},a^{3})+\frac{T^{-1}_{0}}{r}+ O(r^{-1}).
\end{eqnarray}
Once we solve for the optimal isometric embedding equation 
\begin{eqnarray}
 \text{div}_{\sigma}\left(\rho\nabla\tau-\nabla[\sinh^{-1}(\frac{\rho\Delta\tau}{|H||H_{0}|})]-\alpha_{H_{0}}+\alpha_{H}\right)=0,   
\end{eqnarray}
then $\tau:=-\langle X,T_{0}\rangle$ is determined on each $\Sigma_{r}$. Now we solve the isometric embedding equation \ref{eq:isometric} order by order in $r$
which yields
\begin{eqnarray}
\partial_{a}\widehat{X}^{i}\partial_{b}\widehat{X}^{i}=\widehat{\sigma}_{ab},\\
2\partial_{a}\widehat{X}^{i}\partial_{b}(X^{i})^{0}=\sigma^{1}_{ab}.
\end{eqnarray}
The first equation is trivial while the second equation can be solved by means of the following ansatz 
\begin{eqnarray}
(X^{i})^{0}=p^{a}\widehat{\nabla}_{a}\widehat{X}^{i}+v\widehat{X}^{i}
\end{eqnarray}
for $p\in T\mathbb{S}^{2}, v\in C^{\infty}(\mathbb{S}^{2})$. 
Now recall $H_{0}=-\Delta X$ and the inverse metric $\sigma^{ab}=\frac{1}{r^{2}}\widehat{\sigma}^{ab}+O(r^{-3})$. One can find an explicit expansion for $|H_{0}|$
\begin{eqnarray}
|H_{0}|=\frac{2}{r}+\frac{h^{-2}_{0}}{r^{2}}+O(r^{-3})
\end{eqnarray}
where $h^{-2}_{0}$ reads 
\begin{eqnarray}
h^{-2}_{0}=-\widehat{X}^{i}\widehat{\Delta}(X^{i})^{0}-\widehat{\sigma}^{ab}\sigma^{1}_{ab}.
\end{eqnarray}
Now note $\tau=-\langle X,T_{0}\rangle=-ra_{i}\widehat{X}^{i}+O(1)$. Notice that $\widehat{X}^{i}$ lie in the kernel of $\widehat{\Delta}+2$. Now we attempt to solve the optimal isometric embedding equation for $\tau$. The leading order optimal isometric embedding equation is at the level of $O(r^{-3})$ which yields 
\begin{eqnarray}
\label{eq:optimal}
\widehat{\text{div}}(\rho^{-2}\widehat{\nabla}\tau^{1})-\frac{1}{4}\widehat{\Delta}(\rho^{-2}\widehat{\Delta}\tau^{1})-\frac{1}{2}\widehat{\Delta}(\widehat{\Delta}+2)(X^{0})^{0}+\widehat{\text{\div}}(\alpha^{-1}_{H})=0
\end{eqnarray}
where $\alpha_{H}=\frac{\alpha^{-1}_{H}}{r}+\frac{\alpha^{-2}_{H}}{r^{2}}+O(r^{-3})$ and $\alpha_{H_{0}}=\frac{\alpha^{-1}_{H_{0}}}{r}+\frac{\alpha^{-2}_{H_{0}}}{r^{2}}+O(r^{-3})$. An explicit calculation yields 
\begin{eqnarray}
(\alpha^{-1}_{H_{0}})_{a}=\frac{1}{2}\widehat{\nabla}_{a}((\widehat{\Delta}+2)(X^{0})^{0})
\end{eqnarray}
Notice that the data $(|H_{0}|,\alpha_{H_{0}})$ are obtainable from the data $(\sigma,H,\alpha_{H})$ and the isometric embedding $X$. Now we obtain the following lemma 
\begin{lemma}
\label{integrability}
 The following identity is verified for an optimal isometric embedding 
 \begin{eqnarray}
 a^{i}\int_{\mathbb{S}^{2}}(h^{-2}_{0}-h^{-2})=a^{0}\int_{\mathbb{S}^{2}}\widehat{X}^{i}\widehat{\text{div}}(\alpha^{-1}_{H})    
 \end{eqnarray}
\end{lemma}
\begin{proof}
First recall that the optimal isometric embedding equation at $O(r^{-3})$ reads 
\begin{eqnarray}
\label{eq:optimal1}
 \widehat{\text{div}}(\rho^{-2}\widehat{\nabla}\tau^{1})-\frac{1}{4}\widehat{\Delta}(\rho^{-2}\widehat{\Delta}\tau^{1})-\frac{1}{2}\widehat{\Delta}(\widehat{\Delta}+2)(X^{0})^{0}+\widehat{\text{\div}}(\alpha^{-1}_{H})=0   
\end{eqnarray}
where $\rho=\frac{\rho^{-2}}{r^{2}}+\frac{\rho^{-3}}{r^{3}}+o(r^{-3})$ is evaluated following the expression 
\begin{eqnarray}
 \rho=\frac{\sqrt{|\mathbf{H}_{0}|^{2}+\frac{(\Delta\tau)^{2}}{1+|\nabla\tau|^{2}}}-\sqrt{|\mathbf{H}|^{2}+\frac{(\Delta\tau)^{2}}{1+|\nabla\tau|^{2}}}}{\sqrt{1+|\nabla\tau|^{2}}}
\end{eqnarray}
and $\rho^{-2}$ reads 
\begin{eqnarray}
 \rho^{-2}=\frac{h^{-2}_{0}-h^{-2}}{a_{0}}.   
\end{eqnarray}
Now multiply (\ref{eq:optimal1}) by $\widehat{X}^{i}$, integrate by parts over $\mathbb{S}^{2}$, use Stokes together with $\widehat{X}^{i}\in \ker(\widehat{\Delta}+2)$ to yield
\begin{eqnarray}
\int_{\mathbb{S}^{2}}(h^{-2}_{0}-h^{-2})a_{i}\widehat{X}^{i}\widehat{X}^{j}=a^{0}\int_{\mathbb{S}^{2}}\widehat{X}^{j}\widehat{\text{div}}(\alpha^{-1}_{H})
\end{eqnarray}
Now $\widehat{X}^{i}\widehat{X}^{j}=\delta^{ij}$ yielding 
\begin{eqnarray}
a^{i}\int_{\mathbb{S}^{2}}(h^{-2}_{0}-h^{-2})=a^{0}\int_{\mathbb{S}^{2}}\widehat{X}^{i}\widehat{\text{div}}(\alpha^{-1}_{H})
\end{eqnarray}
which completes the proof of the lemma.
\end{proof}
In the next lemma, we prove the sufficient condition for the total quasi-local entity $E(\Sigma,X,T_{0},S)$ to be finite.

\begin{lemma}
\label{finite}
The quasi-local entity $E(\Sigma,X,T_{0},S=X^{\alpha}\partial_{\alpha})$ is finite if \begin{eqnarray}
\int_{\mathbb{S}^{2}}\widehat{X}^{i}(h^{-2}_{0}-h^{-2})\mu_{\mathbb{S}^{2}}=0.
\end{eqnarray}
\end{lemma}
\begin{proof}
First recall the explicit expression for the quasi-local entity $E(\Sigma,X,T_{0},S=X^{\alpha}\partial_{\alpha})$
\begin{eqnarray}
E(\Sigma,X,T_{0},S=X^{\alpha}\partial_{\alpha})=-\frac{1}{8\pi}\int_{\Sigma}\left(\eta_{\alpha\beta}X^{\alpha}T^{\beta}_{0}\rho+X^{\alpha}\partial_{a}X^{\beta}\eta_{\alpha\beta}J^{a}\right)\mu_{\Sigma}
\end{eqnarray}
Now note that the covariant components of the current vector field $J$ verifies 
\begin{eqnarray}
J_{a}=\rho\nabla_{a}\tau-\nabla_{a}[\sinh^{-1}\frac{\rho\Delta \tau}{|H_{0}||H|}]-(\alpha_{H_{0}})_{a}+(\alpha_{H})_{a}=O(r^{-1})    
\end{eqnarray}
 since $\tau=-ra_{i}\widehat{X}^{i}+O(1)$, $\rho=\frac{\rho^{-2}}{r^{2}}+O(r^{-3})=\frac{1}{a_{0}r^{2}}(h^{-2}_{0}-h^{-2})+O(r^{-3})$ etc. We write more explicitly
\begin{eqnarray}
E(\Sigma,X,T_{0},S=X^{\alpha}\partial_{\alpha})=-\frac{1}{8\pi}\int_{\Sigma}\left(-\rho X^{0}T^{0}_{0}+\rho X^{i}T^{i}_{0}-X^{0}\partial_{a}X^{0}J^{a}+X^{i}\partial_{a}X^{i}J^{a}\right)\mu_{\Sigma}.
\end{eqnarray}
Now we understand the scaling of each term 
\begin{eqnarray}
\rho X^{0}T^{0}_{0}\mu_{\Sigma}=O(1)\\
\rho X^{i} T^{i}_{0}\mu_{\Sigma}=O(r)\\
X^{0}\partial_{a}X^{0}J^{a}\mu_{\Sigma}=O(r^{-1})\\
X^{i}\partial_{a}X^{i}J^{a}\mu_{\Sigma}=O(r).
\end{eqnarray}
Therefore, two terms that can lead to potential divergence of the total energy $E(\Sigma,X,T_{0},S=X^{\alpha}\partial_{\alpha})$ are the $O(r)$ terms. Let us first expand $\rho X^{i}T^{i}_{0}\mu_{\Sigma}$ 
\begin{eqnarray}
\rho X^{i}T^{i}_{0}\mu_{\Sigma}=(\frac{h^{-2}_{0}-h^{-2}}{a_{0}r^{2}}+O(r^{-3}))(r\widehat{X}^{i}+(X^{i})^{0}\nonumber+O(r^{-1}))(a^{i}+O(r^{-1}))(r^{2}+O(r))\mu_{\mathbb{S}^{2}}\\\nonumber 
=r\frac{a^{i}}{a^{0}}\widehat{X}^{i}(h^{-2}_{0}-h^{-2})\mu_{\mathbb{S}^{2}}+O(1)=ra^{i}\widehat{X}^{i}\rho^{-2}\mu_{\mathbb{S}^{2}}+O(1)
\end{eqnarray}
where $\mu_{\mathbb{S}^{2}}$ is the volume form on the round 2-unit sphere and $\rho^{-2}$ is explicitly written as $\rho^{-2}=\frac{h^{2}_{0}-h^{-2}}{a_{0}}$. Therefore the first dangerous term reduces to the following 
\begin{eqnarray}
\int_{\Sigma}\rho X^{i}T^{i}_{0}\mu_{\Sigma}=r\frac{a^{i}}{a^{0}}\int_{\Sigma}(h^{-2}_{0}-h^{-2})\mu_{\mathbb{S}^{2}}+O(1).    
\end{eqnarray}
Now we focus on the second dangerous term $X^{i}\partial_{a}X^{i}J^{a}\mu_{\Sigma}$. Writing $J_{a}=\frac{(J^{-1})_{a}}{r}+O(r^{-2})$, we obtain 
\begin{eqnarray}
X^{i}\partial_{a}X^{i}J^{a}\mu_{\Sigma}=r\widehat{X}^{i}\widehat{\nabla}_{a}\widehat{X}^{i}(J^{-1})^{a}\mu_{\mathbb{S}^{2}}+O(1)
\end{eqnarray}
where $(J^{-1})^{a}$ reads 
\begin{eqnarray}
(J^{-1})^{a}=\widehat{\sigma}^{ab}\left(-\rho^{-2}\widehat{\nabla}_{b}\tau^{1}+\frac{1}{4}\widehat{\nabla}_{b}(\rho^{-2}\widehat{\Delta}\tau^{1})-(\alpha^{-1}_{H})_{b}+\frac{1}{2}\widehat{\nabla}_{b}(\widehat{\Delta}+2)(X^{0})^{0}\right).
\end{eqnarray}
But note $\widehat{X}^{i}\widehat{X}^{i}=1$ and therefore 
\begin{eqnarray}
\widehat{X}^{i}\hnabla_{a}\widehat{X}^{i}=0
\end{eqnarray}
yielding $X^{i}\partial_{a}X^{i}J^{a}\mu_{\Sigma}=O(1)$.
Therefore, the only condition that is required for the finiteness of the total quasi-local entity is the following 
\begin{eqnarray}
\int_{\mathbb{S}^{2}}\widehat{X}^{i}(h^{-2}_{0}-h^{-2})\mu_{\mathbb{S}^{2}}=0.
\end{eqnarray}
This concludes the proof of the lemma.
\end{proof}

\subsection{Proof of the theorem 1.1}
To prove the total energy $E(\Sigma,X,T_{0},S=X^{\alpha}\partial_{\alpha})$ (recall that for total energy, we essentially take the limit $r\to\infty$) is finite for asymptotically flat initial data $(M,g,k)$ of order one as defined in \ref{def2}, we only need to prove that the following holds 
\begin{eqnarray}
 \int_{\mathbb{S}^{2}}\widehat{X}^{i}(h^{-2}_{0}-h^{-2})\mu_{\mathbb{S}^{2}}=0.   
\end{eqnarray}
But this follows from lemma 7.2 and lemma 7.3 of \cite{CWY} together with the assumption that $k=O(r^{-2})$. For completeness, we state these two lemma. 
\begin{lemma}[\cite{CWY}]
\label{impt}
Let $(M,g,k)$ be a vacuum asymptotically flat data of order 1 as defined in \ref{def2}. Consider the coordinate spheres $\Sigma_{r}$ and let $\widehat{h}$ be the mean curvature of the coordinate spheres in $M$. Assume $\widehat{h}=\frac{2}{r}+\frac{\widehat{h}^{-2}}{r^{2}}+O(r^{-3})$, then 
\begin{eqnarray}
 \int_{\mathbb{S}^{2}}\widehat{X}^{i}\widehat{h}^{-2}\mu_{\mathbb{S}^{2}}=0.   
\end{eqnarray}
Also let $\Sigma_{r}$ be a family of surfaces with induced metric $\sigma=r^{2}\widetilde{\sigma}+r\sigma^{1}+O(1)$ and $\widehat{h}_{0}=\frac{2}{r}+\frac{\widehat{h}^{-2}_{0}}{r^{2}}+O(r^{-3})$ be the mean curvature of the isometric embedding of $\Sigma_{r}$ into $\mathbb{R}^{3}$, then we have 
\begin{eqnarray}
\int_{\mathbb{S}^{2}}\widehat{X}^{i}\widehat{h}^{-2}_{0}\mu_{\mathbb{S}^{2}}=0.     
\end{eqnarray}
\end{lemma}
These two lemmas are proven using the second variation formula of the are of the coordinate spheres.  Now by lemma 4 of \cite{chen2011evaluating}, we have $\widehat{h}^{-2}_{0}=h^{-2}_{0}$. On the other hand, the data is asymptotically flat of order $1$ and therefore by the definition \ref{def2}, the second fundamental form $k=O(r^{-2})$ and therefore $h^{-2}=\widehat{h}^{-2}$ since the difference is given by the second fundamental form. Putting all these properties together with lemma \ref{impt} we obtain 
\begin{eqnarray}
\int_{\mathbb{S}^{2}}\widehat{X}^{i}(h^{-2}_{0}-h^{-2})\mu_{\mathbb{S}^{2}}=0.      
\end{eqnarray}
Therefore by lemma \ref{finite}, we have that the total quasi-local energy $E(\Sigma,X,T_{0},S=X^{\alpha}\partial_{\alpha})$ is finite.

\section{Time evolution of the total energy $E(\Sigma,X,T_{0},S=X^{\alpha}\partial_{\alpha})$}
In this last section, we want to understand the evolution of the total energy under the vacuum Einstein flow. Consider the vacuum Einstein's evolution equations 
\begin{eqnarray}
\label{eq:evol1}
\partial_{t}g_{ij}=-2Nk_{ij}+(L_{Y}g)_{ij},\\
\label{eq:evol2}
\partial_{t}k_{ij}=-\nabla_{i}\nabla_{j}N+N(R_{ij}+\text{tr}_{g}k k_{ij}-2k_{il}k^{l}_{j})+(L_{Y}k)_{ij}
\end{eqnarray}
with the initial data $g(0)=g,k(0)=k$ verifying the asymptotic flatness of order one condition as stated in definition \ref{def2}. In addition we choose a gauge in which Lapse $N$ and shift $Y$ verifies $N=1+O(r^{-1})$ and $Y=\frac{Y^{-1}}{r}+O(r^{-2})$. We want to understand the evolution of the total energy $E(\Sigma,X,T_{0},S=X^{\alpha}\partial_{\alpha})$ under the Einstein flow $t\mapsto (g(t),k(t),N(t),Y(t))$. Now in order to accomplish this, we need to solve for the isometric embedding equations at instant of time. Following the evolution equations \ref{eq:evol1}-\ref{eq:evol2}, we observe that the following decay condition holds for $\partial_{t}g_{ij}$ and $\partial_{t}k_{ij}$
\begin{eqnarray}
 \partial_{t}g_{ij}=O(r^{-2}), \partial_{t}k_{ij}=O(r^{-3}).   
\end{eqnarray}
This decay together with the decay for lapse and shift $N=1+O(r^{-1})$ and $Y=\frac{Y^{-1}}{r}+O(r^{-2})$ dictates the following asymptotic expansion for the metric, the second fundamental form, and the connection one form of the normal bundle of the coordinate spheres $\Sigma_{r}$
\begin{eqnarray}
\label{eq:exp1}
\sigma_{ab}=r^{2}\widehat{\sigma}_{ab}+r\sigma^{1}_{ab}+\sigma^{0}_{ab}(t)+o(1)\\
|H|=\frac{2}{r}+\frac{h^{-2}}{r^{2}}+\frac{h^{-3}(t)}{r^{3}}+o(r^{-3})\\
\alpha_{H}=\frac{\alpha^{-1}_{H}}{r}+\frac{\alpha^{-2}_{H}(t)}{r^{2}}+o(r^{-2}).
\end{eqnarray}
Now recall we solve for the isometric embedding equations and the optimal isometric embedding equations at instant of time. The isometric embedding $\langle dx,dx\rangle_{\eta}=\sigma$ similarly yields the following asymptotic expansion
\begin{eqnarray}
X^{0}=(X^{0})^{0}+\frac{(X^{0})^{-1}(t)}{r}+o(r^{-1})\\
X^{i}=r\widehat{X}^{i}+(X^{i})^{0}+\frac{(X^{i})^{-1}(t)}{r}+o(r^{-1})\\
T_{0}=(a^{0},a^{i})+\frac{T^{-1}_{0}(t)}{r}+O(r^{-2})
\end{eqnarray}
which in turn yields the following expansion for the mean curvature and the connection $1-$form of the normal bundle of embedded sphere in the Minkowski space
\begin{eqnarray}
\label{eq:exp2}
|H_{0}|=\frac{2}{r}+\frac{h^{-2}_{0}}{r^{2}}+\frac{h^{-3}_{0}(t)}{r^{3}}+O(r^{-4})\\
\alpha_{H_{0}}=\frac{\alpha^{-1}_{H_{0}}}{r}+\frac{\alpha^{-2}_{H_{0}}(t)}{r^{2}}+O(r^{-3}).
\end{eqnarray}
Using these expansions, one obtains a similar expansion for $\tau:=-\langle X,T_{0}\rangle$. Now recall the explicit expression for the energy $E(\Sigma,X,T_{0},S=X^{\alpha}\partial_{\alpha})$
\begin{eqnarray}
E(\Sigma,X,T_{0},S=X^{\alpha}\partial_{\alpha})=-\frac{1}{8\pi}\int_{\Sigma}\left(-\rho X^{0}T^{0}_{0}+\rho X^{i}T^{i}_{0}-X^{0}\nabla_{a}X^{0}J^{a}+X^{i}\nabla_{a}X^{i}J^{a}\right)\mu_{\Sigma}.    
\end{eqnarray}
First, we have the following proposition 
\begin{proposition}
The following holds under the asymptotic decomposition 
\begin{eqnarray}
\partial_{t}\int_{\Sigma}X^{i}\rho \mu_{\Sigma}=\frac{1}{a^{0}}\int_{\mathbb{S}^{2}}\widehat{X}^{i}\partial_{t}\left(h^{-3}_{0}(t)-h^{-3}(t)\right)\mu_{\mathbb{S}^{2}}+O(r^{-1})
\end{eqnarray}
\end{proposition}
\begin{proof}
See First recall the expansions \ref{eq:exp1}-\ref{eq:exp2} and observe $\partial_{t}\int_{\Sigma_{r}}X^{i}\rho\mu_{\Sigma_{r}}=\int_{\Sigma_{r}}X^{i}\partial_{t}\rho\mu_{\Sigma_{r}}+O(r^{-1})$. Now recall the definition of $\rho$
\begin{eqnarray}
\rho=\frac{\sqrt{|\mathbf{H}_{0}|^{2}+\frac{(\Delta\tau)^{2}}{1+|\nabla\tau|^{2}}}-\sqrt{|\mathbf{H}|^{2}+\frac{(\Delta\tau)^{2}}{1+|\nabla\tau|^{2}}}}{\sqrt{1+|\nabla\tau|^{2}}}    
\end{eqnarray}
and therefore 
\begin{eqnarray}
\label{eq:temp}
\partial_{t}\rho=\frac{\partial_{t}(|H_{0}|-|H|)}{a^{0}}-\frac{(|H_{0}|-|H|)\sum_{i}a_{i}\partial_{t}T^{-1}_{0}(t)^{j}}{r(a^{0})^{3}}+O(r^{-4})\\\nonumber 
=\frac{1}{a^{0}}\frac{\partial_{t}\left(h^{-3}_{0}(t)-h^{-3}(t)\right)}{r^{3}}+\frac{h^{-2}_{0}-h^{-2}}{(a^{0})^{3}r^{3}}\sum_{i}a^{j}\partial_{t}T^{-1}_{0}(t)^{i}+O(r^{-4}).
\end{eqnarray}
Now evaluate $\partial_{t}\int_{\Sigma_{r}}X^{i}\rho\mu_{\Sigma_{r}}=\int_{\Sigma_{r}}X^{i}\partial_{t}\rho\mu_{\Sigma_{r}}+O(r^{-1})$ using the expression \ref{eq:temp} 
\begin{eqnarray}
 \partial_{t}\int_{\Sigma_{r}}X^{i}\rho\mu_{\Sigma_{r}}=\frac{1}{a_{0}}\int_{\mathbb{S}^{2}}\widehat{X}^{i}\partial_{t}(h^{-3}_{0}(t)-h^{-3}(t))\mu_{\mathbb{S}^{2}}\nonumber+\frac{\sum_{i}a^{j}\partial_{t}T^{-1}_{0}(t)^{i}}{(a^{0})^{3}}\int_{\mathbb{S}^{2}}\widehat{X}^{i} (h^{-2}_{0}-h^{-2})\mu_{\mathbb{S}^{2}}+O(r^{-1}). 
\end{eqnarray}
Now since the quasi-local energy $E(\Sigma,X,T_{0},S=X^{\alpha}\partial_{\alpha})$ corresponding to the initial data is finite by the asymptotic flatness of order one condition, we have $\int_{\mathbb{S}^{2}}\widehat{X}^{i} (h^{-2}_{0}-h^{-2})\mu_{\mathbb{S}^{2}}=0$ by lemma \ref{finite}. Therefore we have 
\begin{eqnarray}
\partial_{t}\int_{\Sigma}X^{i}\rho \mu_{\Sigma}=\frac{1}{a^{0}}\int_{\mathbb{S}^{2}}\widehat{X}^{i}\partial_{t}\left(h^{-3}_{0}(t)-h^{-3}(t)\right)\mu_{\mathbb{S}^{2}}+O(r^{-1})    
\end{eqnarray}
which concludes the proof of the proposition
\end{proof}
We continue to compute the time derivative of the each of the terms in the expression for the quasi-local entity $E(\Sigma,X,T_{0},S=X^{\alpha}\partial_{\alpha})$ 
\begin{proposition}
Under the asymptotic decomposition, the following estimate holds 
\begin{eqnarray}
\partial_{t}\int_{\Sigma}\rho X^{0}T^{0}_{0}\mu_{\Sigma}=O(r^{-1})
\end{eqnarray}
\end{proposition}
\begin{proof}
This follows by explicit computation. First using the asymptotic expansions \ref{eq:exp1}-\ref{eq:exp2} we have the following 
\begin{eqnarray}
\partial_{t}\rho=\frac{1}{a^{0}}\frac{\partial_{t}\left(h^{-3}_{0}(t)-h^{-3}(t)\right)}{r^{3}}+\frac{h^{-2}_{0}-h^{-2}}{(a^{0})^{3}r^{3}}\sum_{j}a^{j}\partial_{t}T^{-1}_{0}(t)_{j}+O(r^{-4}),\\
\partial_{t}X^{0}=\frac{\partial_{t}(X^{0})^{-1}(t)}{r}+o(r^{-1}),\\
\partial_{t}T^{0}_{0}=\frac{\partial_{t}(T^{0})^{-1}(t)}{r}+O(r^{-2}),\\
\partial_{t}\mu_{\Sigma}=\frac{1}{2}\mu_{\Sigma}\sigma^{ab}\partial_{t}\sigma_{ab}=(r^{2}+O(r))(\frac{1}{r^{2}}\widehat{\sigma}^{ab}+O(r^{-3}))\partial_{t}\sigma^{0}_{ab}(t)=O(1)
\end{eqnarray}
Collecting all the terms yields the result. 
\end{proof}

\begin{proposition}
Under the asymptotic expansion \ref{eq:exp1}-\ref{eq:exp2}, we have 
\begin{eqnarray}
\partial_{t}\int_{\Sigma}\rho X^{i}T^{i}_{0}\mu_{\Sigma}=\frac{1}{a_{0}}\int_{\mathbb{S}^{2}}\partial_{t}(h^{-3}_{0}(t)-h^{-3}(t))a_{i}\widehat{X}^{i}\mu_{\mathbb{S}^{2}}+O(r^{-1})
\end{eqnarray}
\end{proposition}
\begin{proof}
We use the expansion 
\begin{eqnarray}
\partial_{t}\rho=\frac{1}{a^{0}}\frac{\partial_{t}\left(h^{-3}_{0}(t)-h^{-3}(t)\right)}{r^{3}}+\frac{h^{-2}_{0}-h^{-2}}{(a^{0})^{3}r^{3}}\sum_{j}a^{j}\partial_{t}T^{-1}_{0}(t)_{j}+O(r^{-4})
\end{eqnarray}
to yield the $O(1)$ terms $\frac{1}{a_{0}}\int_{\mathbb{S}^{2}}\partial_{t}(h^{-3}_{0}(t)-h^{-3}(t))a_{i}X^{i}\mu_{\mathbb{S}^{2}}$ and a term that is proportional to $\int_{\mathbb{S}^{2}}(h^{-2}_{0}-h^{-2})\widehat{X}^{i}a_{i}\mu_{\mathbb{S}^{2}}$. The second term vanishes following the finiteness property of the total energy of the initial data from lemma \ref{finite}. This completes the proof. 
\end{proof}
Now recall the definitions of the Wang-Yau quasi-local energy and the linear momentum 
\begin{eqnarray}
e=\frac{1}{8\pi}\int_{\Sigma}(\rho T^{0}_{0}+J^{a}\nabla_{a}X^{0})\mu_{\Sigma}\\
p^{i}=\frac{1}{8\pi}\int_{\Sigma}(\rho T^{i}_{0}+J^{a}\nabla_{a}X^{i})\mu_{\Sigma}.    
\end{eqnarray}
and one obtains through the asymptotic expansion 
\begin{eqnarray}
\label{eq:ADM1}
e=\frac{1}{8\pi}\int_{\mathbb{S}^{2}}(h^{-2}_{0}-h^{-2})\mu_{\mathbb{S}^{2}}+O(r^{-1})\\
\label{eq:ADM2}
p^{i}=\frac{1}{8\pi}\int_{\mathbb{S}^{2}}\widehat{X}^{i}\div(\alpha^{-1}_{H})\mu_{\mathbb{S}^{2}}+O(r^{-1})
\end{eqnarray}
and therefore at asymptotic infinity one has $e$ as the ADM energy and $p^{i}$ as the ADM momentum components. Now we invoke the following lemma from \cite{CWY} that are vital to proving the conservation theorem. 
\begin{lemma}[\cite{CWY}]
\label{final}
Under the Einstein evolution equation with lapse function $N=1+O(r^{-1})$ and shift vector field $Y=O(r^{-1})$, we have 
\begin{eqnarray}
\int_{\mathbb{S}^{2}}\widehat{X}^{i}\partial_{t}h^{-3}_{0}\mu_{\mathbb{S}^{2}}=0,\\
\int_{\mathbb{S}^{2}}\widehat{X}^{i}\partial_{t}h^{-3}\mu_{\mathbb{S}^{2}}=-8\pi p^{i}.
\end{eqnarray}
\end{lemma}
Now we can prove the conservation theorem.
\section{Proof of the theorem \ref{theorem2}}
First, collecting all the terms, we obtain 
\begin{eqnarray}
 \partial_{t}E(\Sigma,X,T_{0},K=S=X^{\alpha}\partial_{\alpha})= \frac{1}{8\pi a^{0}}\int_{\mathbb{S}^{2}}\partial_{t}(h^{-3}_{0}(t)-h^{-3}(t))a_{i}\widehat{X}^{i}\mu_{\mathbb{S}^{2}}+O(r^{-1})  
\end{eqnarray}
which by lemma \ref{final} reduces to the following 
\begin{eqnarray}
\partial_{t}E(\Sigma,X,T_{0},K=S=X^{\alpha}\partial_{\alpha})&=& -\frac{1}{8\pi a^{0}}\int_{\mathbb{S}^{2}}h^{-3}(t)a_{i}\widehat{X}^{i}\mu_{\mathbb{S}^{2}}+O(r^{-1}) \\\nonumber 
&=&\frac{p^{i}a_{i}}{a^{0}}+O(r^{-1}).
\end{eqnarray}
Now we want to abe able relate $p^{i}$, $a^{i}$, and $a^{0}$. To this end recall the lemma \ref{integrability}. This is essentially an integrability condition for the optimalisometric embedding equation at the lowest non-trivial order
\begin{eqnarray}
\label{eq:2}
a^{i}\int_{\mathbb{S}^{2}}(h^{-2}_{0}-h^{-2})\mu_{\mathbb{S}^{2}}=a^{0}\int_{\mathbb{S}^{2}}\widehat{X}^{i}\widehat{\text{div}}(\alpha^{(-1)}_{H})\mu_{\mathbb{S}^{2}}.
\end{eqnarray}
Now recall the asymptotic expansion of $e$ and $p^{i}$ in \ref{eq:ADM1}-\ref{eq:ADM2} and observe at spacelike infinity 
\begin{eqnarray}
 \int_{\mathbb{S}^{2}}(h^{-2}_{0}-h^{-2})\mu_{\mathbb{S}^{2}}=8\pi e,~\int_{\mathbb{S}^{2}}\widehat{X}^{i}\widehat{\text{div}}(\alpha^{(-1)}_{H})\mu_{\mathbb{S}^{2}}=8\pi p^{i}   
\end{eqnarray}
and therefore we have at spacelike infinity due to \ref{eq:2} 
\begin{eqnarray}
a^{i}=\frac{a^{0}p^{i}}{e}    
\end{eqnarray}
yielding at asymptotic infinity
\begin{eqnarray}
\partial_{t}E(\Sigma,X,T_{0},K=X^{\alpha}\partial_{\alpha})=\frac{p^{i}p_{i}}{e}=\frac{p^{2}}{e}.
\end{eqnarray}
This concludes the proof of the conservation theorem. 
\section{Conclusion}
We defined a new quasi-local entity by using the definition introduced by \cite{CWY}. This is accomplished by transplanting the conformal Killing vector field the scaling vector field $S=X^{\alpha}\partial_{\alpha}$ of the Minkowksi space through the isometric embedding of the topological $2-$sphere under study. This quasi-local entity verifies several desirable properties. In particular it is finite for initial data verifying asymptotic flatness of order one condition. Moreover, the total energy verifies dynamical law $\partial_{t}E(t)=\frac{p^{2}}{e}$ in terms of the ADM momentum $p$ and ADM energy $e$. This quasi-local entity has potential application in the stability problem of asymptotically flat spacetimes of order one, in particular, black-hole spacetimes. An open problem to study is its positivity property. Note that $S=X^{\alpha}\partial_{\alpha}$ is time-like within the future light cone in Minkowski space. Therefore through a causality argument, one would conjecture that the quasi-local energy defined in the current context is strictly positive and vanishes for a sphere embedded in the Minkowski space. The vanishing property is straightforward to note from the definition while the positivity remains a non-trivial task. This needs to be investigated in the future. Another conformal Killing vector field that is of extreme importance in the stability problem to construct weighted energy is the inversion generator $K=\eta(X,X)\partial_{t}+2X^{0}X^{\alpha}\partial_{\alpha}$. However, in the current context, due to being quadratic in the coordinates, this vector field does not yield a meaningful entity if one simply plugs it in the definition \ref{def1}. One may require an appropriate modification to construct quasi-local conserved charges associated with the remaining conformal Killing vector fields or even perhaps the Killing tensors.

\section{Acknowledgement} This work was supported by the Center of Mathematical Sciences and Applications (CMSA), Department of Mathematics at Harvard University.

\end{document}